\documentclass[a4paper,11pt]{article}
\usepackage[utf8]{inputenc}

\usepackage{amsthm,amsmath,amssymb,amsfonts,graphics, xcolor}

\theoremstyle{plain}
\newtheorem{theorem}{Theorem}
\newtheorem{Pro}[theorem]{Proposition}
\newtheorem{Lem}[theorem]{Lemma}

\newtheorem{Ex}{Example}

\theoremstyle{definition}

\theoremstyle{remark}

\def\Z{\mathbb{Z}}
\def\C{\mathbb{C}}
\def\R{\mathbb{R}}
\def\N{\mathbb{N}}

\def\k{\kappa}
\def\mx{\a_1}
\def\mn{\a_k}
\def\cT{\mathcal{T}}
\def\cS{{\mathcal S}}
\def\bbbn{{\mathbb N}}
\def\bbbz{{\mathbb Z}}

\def\i{{\rm i}}

\renewcommand{\leq}{\leqslant}
\renewcommand{\geq}{\geqslant}

\def\bbbc{{\mathbb C}}
\newcommand\cD{{\mathcal D}}

\def\a{\alpha}

\def\om{\omega}

\def\d{\partial}
\def\fieldk{\bbbc}

\newcommand\cB{{\mathcal B}}
\newcommand\cU{{\mathcal U}}

\def\cM{\mathcal{M}}
\def\cN{\mathcal{N}}
\def\fA{\mathfrak{A}}

\def\ff{\mathfrak{f}}
\def\cA{\mathcal{A}}
\def\cT{\mathcal{T}}
\def\cS{{\mathcal S}}
\def\cZ{\mathcal{Z}}

\def\fI{\mathfrak{I}}

\def\pI{\pi_{\fI_a}}
\def\pJ{\pi_{\fI_b}}
\def\cI{\mathcal{I}}

\begin{document}

\bibliographystyle{unsrt}
\title{Hamiltonians for the quantised Volterra hierarchy}
\author{Sylvain Carpentier $^\ddagger$, Alexander V. Mikhailov$^{\star}$ and
Jing Ping
Wang$ ^\dagger $
\\
$\ddagger$ QSMS, Seoul National University, South Korea,
sylcar@snu.ac.kr
\\
$\star$ School of Mathematics, University of Leeds, UK,
a.v.mikhailov@leeds.ac.uk\\
$\dagger$ School of Mathematics, Statistics \& Actuarial Science, University of
Kent, UK, J.Wang@kent.ac.uk
}
\date{}

\maketitle

\begin{abstract}
This paper builds upon our recent work, published  in {\em Lett
Math Phys, 112:94, 2022}, where we established that the integrable Volterra lattice on a free associative algebra and the whole hierarchy of its symmetries admits a quantisation dependent on a  parameter $\omega$. We also uncovered an intriguing aspect: all odd-degree symmetries of the hierarchy admits an alternative, non-deformation quantisation, resulting in a non-commutative  algebra for any choice of the quantisation parameter $\omega$. In this study, we demonstrate that each equation within the quantum Volterra hierarchy can be expressed in the Heisenberg form. We provide explicit expressions for all quantum Hamiltonians and establish their commutativity.  In the classical limit, these quantum Hamiltonians yield explicit expressions for the classical ones of the commutative Volterra hierarchy. Furthermore, we present Heisenberg equations and their Hamiltonians in the case of  non-deformation quantisation. Finally, we discuss commuting first integrals, central elements of the quantum algebra, and the integrability problem for periodic reductions of the Volterra lattice in the context of both quantisations.
\end{abstract}

\section{Introduction}

In this paper we develop further the quantisation theory of the Volterra 
lattice \cite{CMW}, based on the notion of quantisation ideals \cite{AvM20}. In 
our paper \cite{CMW} we proved that the non-Abelian Volterra lattice, together 
with the whole hierarchy of its commuting symmetries, admits a quantisation 
with quadratic commutation relations between dynamical variables that depends 
on a complex parameter $\omega$. The algebra generated by the dynamical 
variables becomes commutative in the specialisation $\omega=1$. This 
quantisation can be viewed as a finite deformation of a commutative algebra, a 
deformation that is consistent with all equations of the hierarchy. Slightly
abusing terminology, we shall call it {\em deformation quantisation}, although 
it does not use to any Poisson structure of the
Volterra lattice and the noncommutative multiplication is presented in an
explicit form in contrast to  the well known theory of deformation quantisation 
\cite{Flato, Kontsevich}. In addition, we also showed that all
odd-degree symmetries of the Volterra hierarchy admit a non-deformation 
quantisation whose multiplication law is noncommutative for any choice of the 
quantisation parameter $\omega$ \cite{CMW}. While the deformation quantisation 
for the Volterra lattice is known in the literature \cite{InKa}, the 
non-deformation quantisation appeared in \cite{AvM20} for the first time. In
physics, a quantum description of fermions can be regarded as non-deformation 
quantisation, since the ``classical'' limit of the fermion dynamical variables 
is represented by a $\Z_2$ graded (Grassmann) algebra with commutative and 
anti-commutative variables. The ``classical'' limit of the 
Volterra non-deformation quantum algebra is not commutative and is not graded. 
It is a new type of non-commutative associative algebras whose representation 
theory has not yet been developed.

Traditionally, commuting quantum integrals are obtained in the frame of
the quantum inverse scattering method using a Lax representation with
ultra-local $L$ operator \cite{InKa, volkov, babel}. In this
method, the quantum commuting operators, including the Hamiltonian of the 
system, can be obtained recursively from the logarithm of the trace of the
monodromy matrix (the transfer matrix) using it as a generating function.
The coefficients in the expansion of this generating function
commute thanks to the existence of a quantum $R$ matrix compatible with the ultra-local $L$ operators.
The goal of this paper is to present the Hamiltonian operators
in explicit form and show that they are formally self-adjoint,
commute with each other, and
yield Heisenberg equations for every member of the integrable hierarchy of the
commuting symmetries. We show it without making use of the quantum Lax structure or the
corresponding transfer matrix. We prove this result both for the conventional
and the non-deformation quantisations.

The notion of quantisation ideals for dynamical systems defined on free
algebras was proposed in
\cite{AvM20}. Let $\fA$ be a
free associative algebra with a finite or infinite number of
multiplicative generators. The dynamical system defines a derivation
$\partial_t:\fA\mapsto\fA$.
A quantisation is a canonical projection of the dynamical
system on $\fA$ to a system defined on a quotient algebra
$\fA_\fI=\fA\diagup\fI$ over a two-sided ideal
$\fI\subset\fA$ satisfying the following properties:
\begin{enumerate}
 \item[({\rm i})] the ideal $\fI$ is $\partial_t$--stable, that is, 
$\partial_t(\fI)\subset\fI$;
 \item[({\rm ii})] the quotient algebra $\fA_\fI$ admits an additive basis of 
normally ordered monomials.
\end{enumerate}
An ideal satisfying the above two conditions is called  a
{\em
quantisation ideal}, and $\fA_\fI$ is called a {\em quantum algebra}.

In \cite{CMW}, we applied this approach to the integrable
nonabelian Volterra lattice
\begin{equation}\label{vol}
\partial_{t_1 }(u_n)=  K^{(1)}(u_{n+1},u_n,u_{n-1}), \quad K^{(1)}=u_{n+1} u_n- 
u_n
u_{n-1},\qquad n\in\bbbz
\end{equation}
and its hierarchy of symmetries
\begin{equation}\label{hvol}
\partial_{t_\ell}(u_n)=K^{(\ell)}(u_{n+\ell},\ldots ,u_{n-\ell}),\qquad
\ell\in\bbbn,\quad  n\in\bbbz,
\end{equation}
where $K^{(\ell)}(u_{n+\ell},\ldots ,u_{n-\ell})$ are homogeneous polynomials 
of degree
$\ell+1$ (explicitly given in section \ref{sec21}).
The second member of
the hierarchy
\begin{equation}
\label{voltf2}
\partial_{t_2}(u_n)=K^{(2)}=u_{n+2}  u_{n+1}  u_n +u_{n+1}^2  u_n+  u_{n+1}  
u_n^2 -
u_n^2  u_{n-1}-u_n   u_{n-1}^2- u_n   u_{n-1}  u_{n-2}
\end{equation}
is a cubic polynomial and we refer to it as the cubic symmetry of (\ref{vol}). 
In this case the free algebra $\fA=\C[\omega]\langle u_n\ ;\ n\in\Z\rangle$ is 
generated by infinite number of non-commutative variables $u_n$.
We proved that the Volterra lattice (\ref{vol}) and its whole hierarchy of
symmetries admit  a 
quantisation with the quantization ideal
\begin{equation}\label{idi}
\fI_a= \langle \{ u_nu_{n+1}-\omega u_{n+1}u_n\,;\ n \in \mathbb{Z} \} \cup
\{u_nu_m-u_mu_n\,;\ |n-m| >1 ,\ n,m \in\bbbz\ \} \rangle ,
 \end{equation}
leading to the commutation relations
\begin{equation}\label{comm1}
  u_nu_{n+1}=\omega u_{n+1}u_n,\qquad u_nu_m=u_mu_n\ \
\mbox{if}\ \ |n-m|\geqslant 2,\quad n,m \in\bbbz
\end{equation}
in the quotient algebra $\fA_{\fI_a}$, where $\omega\in\C^*$ is a quantisation 
parameter. Moreover, we showed that
the cubic symmetry of the Volterra lattice, equation
(\ref{voltf2}), and  all odd degree members of the Volterra
hierarchy also admit a
non-deformation quantisation with the quantisation ideal 
\begin{equation}\label{idj}
\fI_b= \langle \{ u_nu_{n+1}-(-1)^n \omega u_{n+1}u_n\,;\,n \in \mathbb{Z} \}
\cup
\{u_nu_m+u_mu_n\,;\,|n-m| >1,\  n,m \in \mathbb{Z}\} \rangle\,
\end{equation}
and commutation relations
\begin{equation}\label{comm2}
 u_nu_{n+1}=(-1)^n \omega u_{n+1}u_n,\qquad
u_nu_m+u_mu_n=0\ \
\mbox{if}\ \ |n-m|\geqslant 2,\quad n,m \in\bbbz
\end{equation}
in the quotient algebra $\fA_{\fI_b}$.

In the quantum theory, real valued dynamical variables   are replaced
by self-adjoint operators, with respect to a Hermitian conjugation $\dagger$.
The ideals $\fI_a$ and $\fI_b$  and corresponding commutation relations are 
stable with respect to the Hermitian conjugation $\dagger$ (defined in Section 
\ref{sec3}), assuming 
the variables $u_n$ are self-adjoint and $\omega=e^{ \i\hbar}$. Here
$\hbar$ is an arbitrary real parameter, an analogue of the Plank constant, and
$\i=\sqrt{-1}$.
For the quantised equations of the Volterra hierarchy, we introduce the factors 
$e^{\frac12\i\ell\hbar}$ to make
 the right-hand side of the equations self-adjoint, that is,
\begin{equation}\label{qVh}
\d_{t_\ell}(u_n)=q^{\ell}K^{(\ell)}(u_{n+\ell},\ldots ,u_{n-\ell}),\qquad 
q=e^{\frac12\i\hbar}, \
\ell=1,2,\ldots,\ \ n\in\bbbz .
\end{equation}

In this paper we show that the infinite sequence of quantum Hamiltonians 
$H_\ell$ for the quantised Volterra hierarchy defined on the quantum algebra 
$\fA_{\fI_a}$ is given by
\[
 H_\ell=\sum_{k\in\bbbz}\ 
\sum_{\alpha\in\cN^\ell}\frac{\omega^\ell-1}{\omega^{\nu(\alpha,0)}-1}P_\alpha^{
\fI_a}(\omega)u_{\alpha+k},
\]
where the set 
$$\cN^\ell=\{\alpha=(\a_1, \a_2 \cdots, \a_{\ell-1}, 
0)\in\bbbz^\ell\Big{|} 
\a_i=\sum_{s=i}^{\ell-1} \theta_s, \ \theta_s\in\{0,1\}; i=1, \cdots,
\ell-1\},$$ 
and $\nu(\a, i)$ denotes
the number of $i$'s in the $\ell$-tuple $\a$. For $\a\in\cN^\ell$, the 
polynomials
$P_\alpha^{\fI_a}(\omega)$ are given by   products of Gaussian binomials
$$
P_\alpha^{\fI_a}(\omega)=\binom{\nu(\a,\a_1)+\nu(\a,\a_1-1)-1}
{\nu(\a,\a_1)}_\omega \cdots \binom{\nu(\a,1)+\nu(\a,0)-1}{\nu(\a,1)}_\omega\ .
$$
We   prove that the Hamiltonians $H_\ell$ are self-adjoint 
$H_\ell^\dagger=H_\ell$ and  commute with each other $[H_\ell,H_k]=0,\ 
k,\ell\in \N $. Furthermore, the 
dynamical equations of the quantum hierarchy (\ref{qVh}) can be written in the 
Heisenberg 
form \cite{dirac_book}:
\[ 
\partial_{t_\ell}(u_n)=\frac{\i}{2\sin\left(\frac12\ell 
\hbar\right)}[H_\ell,u_n] ,
\qquad n\in\bbbz,\ \ell\in\bbbn\, .
\]
In the classical limit $\hbar\to 0$ we obtain the Volterra hierarchy in the 
Hamiltonian form $\partial_{t_\ell}(u_n)=\{u_n,\tilde H_\ell\} $ and 
explicit expressions for all Hamiltonians $\tilde H_\ell=\lim\limits_{\hbar\to 
0}\ell^{-1}H_\ell$ (see Section \ref{discussion}). 

In the case of the non-deformation quantisation (\ref{idj})  we have also
found explicit expressions for self-adjoint commuting quantum 
Hamiltonians 
and present the quantum hierarchy with even times in the Heisenberg form.   
These results are stated in Theorem \ref{thm2}. 

The problem of quantisation of the Volterra lattice has a long history. In
1992, using the quantum version of the inverse spectral transform method, Volkov
proposed quantum commutation relations between the dynamical
variables \cite{volkov} (see also \cite{babel}). These commutation relations 
are 
hardly suitable for the derivation of the Heisenberg equations and the study of 
the corresponding quantum algebra structure. In the paper by Inoue and Hikami 
\cite{InKa},  the commutation relations (\ref{comm1}), as well as the first 
four Hamiltonians of the quantum Volterra hierarchy were found using ultra-local Lax
representation and the $R$--matrix technique.
Our alternative
approach does not rely on the existence of an ultra-local Lax representation,
$R$--matrix or Hamiltonian structures. It enables us to present   explicitly
 all quantum Hamiltonians for the Volterra hierarchy in the case of 
the deformation quantisation (\ref{comm1}). Moreover, we are able to find
explicitly the Hamiltonians and Heisenberg equations for the non-deformation
quantisation (\ref{comm2}) defined for all odd-degree members of the Volterra hierarchy.
Both results are new and rather surprising.

\section{The Nonabelian Volterra hierarchy and its quantisations}\label{sec2}

In this section, we derive the explicit expressions for the quantised Volterra 
hierarchy
under the quantisation ideal $\fI_a$ defined by \eqref{idi}, making use of 
Gaussian binomial coefficients.
When $\omega=1$, this also reduces to the hierarchy of symmetries for the 
classical (abelian) Volterra chain.
We first give a brief description of the nonabelian Volterra
hierarchy and introduce some basic notations required for this paper.

\subsection{The Nonabelian Volterra hierarchy}\label{sec21}
Let $\fA=\fieldk\langle u_n\,;\, n\in\bbbz\rangle$ be the free associative
algebra of polynomials generated by an infinite number of non-commuting 
variables $u_n$.
There is a natural automorphism $\cS\,:\,\fA\mapsto\fA$,
which we call the {\em shift operator}, defined by
\[
 \cS: a( u _k,\ldots , u _r)\mapsto a( u _{k+1},\ldots , u _{r+1}),\quad
\cS:\alpha\mapsto\alpha, \qquad a( u _k,\ldots , u _r)\in \fA,\ \
\alpha\in\fieldk.
\]
Thus $(\fA, \cS)$ is a difference algebra.
A derivation  $\cD  $ of the algebra  $\fA$ is a $\C$--linear map
satisfying   Leibniz's rule
\[\cD  (\alpha a+\beta  b)=\alpha\cD  (a)+\beta\cD  (b),\qquad  \cD  (a\cdot
b)=\cD (a)\cdot b+a\cdot\cD (b),\qquad  a,b\in\fA,\ \ \alpha,\beta\in\C.\]
It is uniquely defined by its action on the
generators
and $\cD  (\alpha)=0,\ \alpha\in\C$.

A derivation $\cD $ is called evolutionary if it commutes with the shift 
operator
$\cS$. An evolutionary derivation is completely characterised by its action on
the generator $u$ (we often write $u$ instead of $u_0$), that is,
\[
 \cD  (u)=a\quad \mbox{and} \quad \cD (u_k)=\cS^k (a),\qquad a\in \fA.
\]
We adopt the notation $\cD_a$ for the unique evolutionary derivation of $\fA$ 
such that $\cD_a(u)=a$. Evolutionary
derivations form a Lie subalgebra of the Lie algebra  of derivations of $\fA$, 
and the characteristic of a commutator $[\cD_a,\cD_b]=\cD_c$ is given by 
$c=\cD_a(b)-\cD_b(a)$. This expression induces a Lie bracket on the difference 
algebra $\fA$.

Assuming that the generators $u_k$ depend on $t\in\C$ we then identify the
evolutionary derivation $\cD_a$ with an infinite system of 
differential-difference
equations
\[
 \d_t(u_n)=\cD_a(u_n)=\cS^n(a).\qquad n\in\Z.
\]
From now on we will think of the system of evolutionary equations and the 
evolutionary derivation as the same object.

The Volterra lattice (\ref{vol}) defines an evolutionary derivation
$\d_{t_1}\,:\,\fA\mapsto\fA.$
The differential-difference system (\ref{voltf2}) defines another evolutionary
derivation $\d_{t_2}$.
Evolutionary derivations commuting with $\d_{t_1}$ are called (generalised) 
symmetries of the Volterra lattice. It can be straightforwardly verified that $[\d_{t_1},\d_{t_2}]=0$ and
thus equation (\ref{voltf2}) is a symmetry of the Volterra lattice.

It is well known that the Volterra lattice has an infinite hierarchy of
commuting symmetries.  They can be found using Lax representations both in
commutative \cite{Manakov74} and noncommutative \cite{Bog91} cases, or
using recursion operators \cite{wang12, cw19-2}. Remarkably, the  symmetries of 
the Volterra lattice (\ref{vol}) can be explicitly presented
in terms of a family of nonabelian homogeneous difference polynomials
\cite{cw19-2}, which was inspired by
the family of polynomials discovered in the commutative case (see 
\cite{svin09,svin11}).

Let us assume that the generators  $u_k$ of the free associative algebra $\fA$
depend on an infinite set of ``times'' $t_1,t_2, \ldots$.
It follows from
\cite{cw19-2} that the hierarchy of commuting symmetries of the nonabelian 
Volterra lattice
(\ref{vol}) can be written in the following explicit form
\begin{equation}\label{volh}
 \d_{t_\ell}(u)=\cS(X^{(\ell)})u-u\cS^{-1}(X^{(\ell)}),\qquad \ell\in\N\, ,
\end{equation}
where the (noncommutative) polynomials $X^{(\ell)}$ are given by
\begin{equation}\label{xl}
 X^{(\ell)}=
 \sum_{0\leq \lambda_{1}\leq \cdots \leq \lambda_{\ell}\leq \ell-1}
	\left(\prod_{j=1}^{\rightarrow{\ell}} u_{\lambda_{j}+1-j} \right).
\end{equation}
Here $\prod_{j=1}^{\rightarrow{\ell}}$   denotes the order of the values $j$,
from $1$ to $\ell$ in the product of the noncommutative generators
$u_{\lambda_{j}+1-j}$.
For example, we have $X^{(1)}=u$ and
\begin{eqnarray}
 &&\hspace{-1cm}X^{(2)}=u_1 u +u^2+u u_{-1};\label{x1}\\
 &&\hspace{-1cm} X^{(3)}=u_2 u_1 u+u_1^2 u + u u_1 u +u_1 u^2+u^3+ u u_{-1}
u+u_1 u u_{-1}+u^2 u_{-1}+uu_{-1}^2+u u_{-1} u_{-2};\label{x2}\\
&&\hspace{-1cm} X^{(4)}=u_3 u_2 u_1 u+u_2^2 u_1 u+u_2 u_1^2 u+u_1 u_2 u_1 u+u_2 
u_1 u^2+u u_2 u_1 u+u_2 u_1 u u_{-1}+u_1^2 u^2+u_1 u u_1 u\nonumber\\
&&+u u_1^2 u+u_1^3 u+u u_1 u^2+u_1 u^3+u^2 u_1 u+u_1^2 u u_{-1}+u^4+u u_1 u 
u_{-1}+u_1 u u_{-1} u+u_1 u^2 u_{-1}
\nonumber\\
&&+u u_{-1}^2 u+u u_{-1} u u_{-1}+u u_{-1} u^2+u u_{-1} u_1 u+u^2 u_{-1}^2+u^2 
u_{-1} u+u^3 u_{-1}+u_1 u u_{-1}^2+u u_{-1}^3\label{x3}
\\&&+u_1 u u_{-1} u_{-2}+u u_{-1} u_{-2} u+u^2 u_{-1} u_{-2}+u u_{-1}
u_{-2} u_{-1}+u u_{-1}^2 u_{-2}+u u_{-1} u_{-2}^2+u u_{-1} u_{-2} u_{-3} 
.\nonumber
\end{eqnarray}
Clearly, we get the Volterra equation (\ref{vol}) when $\ell=1$ and the system
(\ref{voltf2}) when $\ell=2$.

Let $\a=(\a_1, \a_2, \cdots, \a_k) \in \bbbz^k$ be a $k$-component vector. For
each $\a\in \bbbz^k$, we define the $k$-degree monomial $u_{\a} =u_{\a_1}
u_{\a_2}\cdots u_{\a_k}$.  We denote the degree of $\a$ by  $|\a|=k$. We say that a monomial $u_\a$ is normally ordered  if $\a_i>\a_{i+1}$ for all $1\leqslant i\leqslant k-1$.
Conventionally, we write $(\a_1+1,\a_2+1, \cdots, \a_k+1)$ as $\a+1$. Thus we
have
 $\cS^i u_{\a}=u_{\a+i}$ for $i\in\bbbz$. The multiplicity of $u_i$ in
the monomial $u_{\a}$ is denoted by $\nu(\a,i)$.
 Similarly, we denote by $\nu(\a,\geq i)$ the number of $k \geq i$ such that
$u_k$ appears in $u_{\a}$, counted with multiplicities.
 We say that two monomials $u_{\a}$ and $u_{\beta}$ are \textit{similar}
written as ${\a} \sim {\beta}$ if $\nu(\a,i)=\nu(\beta,i)$ for all $i\in\bbbz$.

We define two sets of distinguished monomials, namely, \textit{admissible} and
\textit{nonincreasing} monomials.  For $k \geq 1$, let
 \begin{eqnarray}
  && \cA^k=\left\{\a\in \bbbz^k\big{|} \,  1-j\leq \a_{j} \leq k-j, \,
j=1,\cdots, k;
\ \a_{i+1}+1 \geq \a_i,\ i=1,\cdots,k-1\right\};\label{seta}\\
  &&\cZ^k_{\geq}=\left\{\a\in \bbbz^k\big{|} \a_{i+1}+1 \geq \a_i\geq
\a_{i+1},\ i=1,...,k-1\right\}.\label{setz}
 \end{eqnarray}
A $k$-degree monomial $u_{\a}$ is admissible if $\alpha\in \cA^k$
and is nonincreasing if $\a\in \cZ^k_{\geq}$.

Using these notations, the expression $X^{(k)}$ given by
(\ref{xl}) can be written as
\begin{equation}\label{xk}
 X^{(k)}=\sum_{\a\in \cA^k} u_{\a} .
\end{equation}
In what follows, we use this form to present $X^{(k)}$ in normal ordering under
the quantisation ideals of the Volterra
hierarchy and to derive Hamiltonians for their quantised equations.

\subsection{The quantised Volterra hierarchies in normal ordering}\label{sec22}
Assume that $\fI\subset\fA$ is a two-sided ideal generated
by the infinite set of
polynomials $\ff_{i,j}$:
\begin{equation}\label{ideal00}
 \fI=\langle \ff_{i,j}\,;\, i<j,\ i,j\in\Z\rangle,\qquad
 \, \ff_{i,j}=u_i u_j-\omega_{i,j}u_j u_i,
\end{equation}
where $\omega_{i,j}\in\C^*$ are arbitrary non-zero complex parameters.
Specifying the nonzero constants $\omega_{i,j}$ leads to either $\fI_a$ defined 
by \eqref{idi}
or $\fI_b$ defined by \eqref{idj}.

Given such an ideal $\fI$, we denote the projection on the quotient algebra $\fA_{\fI}$ by
$\pi_{\fI}: \fA \rightarrow \fA_{\fI}$.
The algebra $\fA_{\fI}$ has an additive basis of  {\em standard
normally ordered monomials}
\[
 u_{i_1}u_{i_2}\cdots u_{i_n}\, ;\qquad i_1\geqslant
i_2\geqslant\cdots\geqslant i_n,\ i_k\in\Z,\ n\in\N .
\]
The canonical projection
$\pi_{\fI}: \fA \rightarrow \fA_{\fI} $ acts on the polynomial $X^{(k)}$ given 
by \eqref{xk} as
follows:
\begin{eqnarray}\label{ixk}
 \pi_{\fI} (X^{(k)})=\sum_{\a\in \cA^k\cap \cZ^k_{\geq}} P^{\fI}_{\a}(\omega)
u_{\a},
\end{eqnarray}
where $P^{\fI}_{\a}(\omega)$ is the unique polynomial in $\mathbb{Z}[\omega]$
such that for $\a\in  \cA^k\cap \cZ^k_{\geq}$,
 \begin{equation}\label{pa}
 P^{\fI}_{\a}(\omega)u_{\a} =\pi_{\fI}  \left( \sum_{\beta\in \cA^k, \beta \sim
\a} u_{\beta}\right).
 \end{equation}
We often write it as $P_{\a}(\omega)$ if there is no ambiguity regarding the 
choice of ideal.

We now study the polynomials $P_{\a}(\omega)$ for the
quantisation ideals $\fI_a$ \eqref{idi} and $\fI_b$ \eqref{idj}.
For example, we have
\begin{eqnarray} \label{paa}
&& \pI(X^{(1)})=X^{(1)}=u; \qquad \pI(X^{(2)})=X^{(2)}=u_1 u+u^2+u u_{-1};\\
&&\pI(X^{(3)})=u_2 u_1 u+u_1^2 u +(1+ \om) u_1 u^2+u^3+(1+ \om) u^2 u_{-1} +u_1
u u_{-1}+uu_{-1}^2+u u_{-1} u_{-2}; \nonumber\\
&&\pI(X^{(4)})=u_3 u_2 u_1 u+u_2^2 u_1 u+u_1^3 u
+u^4+u_1^2 u u_{-1}+u_2 u_1 u u_{-1}+u_1 u u_{-1}^2+u_1 u u_{-1}
u_{-2}\nonumber\\
&&\qquad
+u u_{-1}^3+u u_{-1} u_{-2}^2+u u_{-1} u_{-2} u_{-3}
+(1+\om) \left(u_2 u_1^2 u+ u_2 u_1 u^2+u^2 u_{-1} u_{-2}+u u_{-1}^2
u_{-2}\right)\nonumber\\
&&\qquad +(1+\om)^2 u_1 u^2 u_{-1}
+(1+\om+\om^2)\left(u_1^2 u^2 +u_1 u^3+u^2 u_{-1}^2+u^3 u_{-1}\right)
\nonumber
\end{eqnarray}
and
\begin{eqnarray}\label{pab}
&& \pJ(X^{(1)})=X^{(1)}=u; \qquad \pJ(X^{(2)})=X^{(2)}=u_1 u+u^2+u u_{-1};\\
&&\pJ(X^{(3)})=u_2 u_1 u+u_1^2 u +(1+\om) u_1 u^2+u^3+(1- \om) u^2 u_{-1} +u_1
u u_{-1}+uu_{-1}^2+u u_{-1} u_{-2}; \nonumber\\
&&\pJ(X^{(4)})=u_3 u_2 u_1 u+u_2^2 u_1 u+u_2 u_1 u u_{-1}+u_1^3 u+u_1^2 u 
u_{-1}+u^4+u_1 u u_{-1}^2
+u_1 u u_{-1} u_{-2}
\nonumber\\&&\qquad +u u_{-1}^3+u u_{-1} u_{-2}^2+u u_{-1} u_{-2} u_{-3}
+(1-\om) \left(u_2 u_1^2 u+u_2 u_1 u^2\right)+(1+\om^2) u_1 u^2 u_{-1}
\nonumber\\&&\qquad
+(1+\om) \left(u^2 u_{-1} u_{-2}+u u_{-1}^2 u_{-2}\right)
+(1+\om+\om^2) \left( u_1^2 u^2+ u_1 u^3\right)
\nonumber\\&&\qquad+(1-\om+\om^2) \left(u^2 u_{-1}^2+u^3 u_{-1} 
\right).\nonumber
\end{eqnarray}
This defines the polynomials $P_{\a}(\om)$ for all $\alpha$ that are
admissible, nonincreasing and of degree $1$ to $4$. For example, 
$P_{(0,0,-1)}^{\fI_a}(\om)=1+\om$ and $P_{(0,0,-1)}^{\fI_b}(\om)=1-\om$.

For the quantisation ideal $\fI_a$, these polynomials can be computed explicitly 
using the
Gaussian binomial coefficients:
$$
\binom{m}{r}_{\om}=\frac{(1-\om^m)(1-\om^{m-1}) \cdots (1-\om^{m-r+1})}{(1-\om) 
(1-\om^2) \cdots (1-\om^r)},
$$
where $m$ and $r$ are non-negative integers.  If $r > m$, this equals zero. 
When 
$r = 0$, its value is $1$.
\begin{Pro}\label{prop1}
For $\a=(\a_1, \cdots, \a_k) \in \cA^k\cap \cZ^k_{\geq}$, let $\k_i=\nu(\a,i)$, 
where $\mn\leq i\leq \mx$. Then
\begin{eqnarray}
&&P_{\a}^{\fI_a}(\om) = \binom{\k_{\mx}+\k_{\mx-1}-1}{\k_{\mx}}_{\omega}...  
\binom{\k_2+\k_1-1}{\k_2}_{\omega}  \binom{\k_1+\k_0-1}{\k_1}_{\omega}
\nonumber\\&&\qquad \qquad
\binom{\k_0+\k_{-1}-1}{\k_{-1}}_{\omega} \cdots 
\binom{\k_{\mn+1}+\k_{\mn}-1}{\k_{\mn}}_{\omega}. \label{binom}
\end{eqnarray}
\end{Pro}
\begin{proof}
An admissible monomial similar to $\a$ is equivalent to the following data:
\begin{enumerate}
    \item [$(\rm i).$] for each integer $i$ such that $0 \leq i \leq \a_1-1$, a
monomial similar to $u_{i+1}^{\k_{i+1}}u_i^{\k_i}$ ending on the right by $u_i$,
    \item[$(\rm ii).$] for each integer $i$ such that $0 \geq i \geq \a_k+1$, a
monomial similar to $u_{i}^{\k_{i}}u_{i-1}^{\k_{i-1}}$ starting on the left by 
$u_i$.
\end{enumerate}
This is true since we have $u_nu_m=u_mu_n$ for $|n-m| >1$ in the quantised 
algebra $\fA_{\fI_a}$. We now need to compute the sums of monomials in $(\rm i)$ for
fixed $i$. Let us denote by $(n|m)$ the monomial $u_{i+1}^nu_i^m$ and by 
$Q_{(n|m)}(\omega)$ the coefficient in front of $u_{i+1}^nu_i^m$ when summing 
all monomials in $(\rm i)$. If $m=1$ then we have $Q_{(n|1)}(\omega)=1$ since the
only monomial similar to $(n|1)$ and ending by $u_i$ is itself. We also have 
$Q_{(0|m)}(\omega)=1$. Otherwise, such admissible monomials start either with 
$u_i$ or $u_{i+1}$ that gives the induction formula
$Q_{(n|m)}=Q_{(n-1|m)}+ \omega^n Q_{(n|m-1)}.$ This can be integrated into
$$Q_{(n|m)}(\omega)=\binom{n+m-1}{n}_{\omega}.$$
Using a mirror argument, one sees that the sum of all monomials in $(\rm ii)$ is
equal 
to
$$\binom{\k_i+\k_{i-1}-1}{\k_{i-1}}_{\omega} u_{i}^{\k_{i}}u_{i-1}^{\k_{i-1}} 
,$$
which concludes the proof.
\end{proof}
It follows from this proposition that
\begin{eqnarray}
&&P_{\a}^{\fI_a}(\om)+ \om^{\nu(\a,0)}P_{\a-1}^{\fI_a}(\om)
=P_{\a-1}^{\fI_a}(\om)+ \om^{\nu(\a,1)}
P_{\a}^{\fI_a}(\om), \quad \a\in \cZ^k_{\geq},\label{eqp}
\end{eqnarray}
which was alternatively proved based on combinatoric counting  in \cite{CMW} 
when we showed
that the ideal $\fI_a$ defined by \eqref{idi} is
preserved by the symmetry flows  \eqref{volh}, for all $\ell\in \bbbn$.

Using Proposition \ref{prop1}, one can directly compute the canonical 
projections under
the quantisation ideal $\fI_a$
without first writing down $X^{(l)}$. For example,
$$
P_{(1,0,0,-1)}^{\fI_a}(\om)=\binom{2}{1}_\om \binom{2}{1}_\om=(1+\om)^2,
$$
which is the coefficient of $u_1 u^2 u_{-1}$ in $\pI(X^{(4)})$ as shown in 
(\ref{paa}).

For the quantisation ideal $\fI_b$, we have not been able to obtain such neat 
formula since
an admissible monomial is not the product of canonical projections of monomials 
$u_i^{\k_i}
u_{i-1}^{\k_{i-1}}$ due
to the relation $u_nu_m+u_mu_n=0$ for $|n-m|>1$ in the quantum algebra 
$\fA_{\fI_b}$. In \cite{CMW}, we proved the following important identity:
\begin{eqnarray}
\hspace{-0.6cm} P_{\a}^{\fI_b}(\om)+ (-1)^{\nu(\a,\geq
0)}\om^{\nu(\a,0)}P_{\a-1}^{\fI_b}(-\om)
=P_{\a-1}^{\fI_b}(-\om)+ (-1)^{\nu(\a,\geq 2)}\om^{\nu(\a,1)} 
P_{\a}^{\fI_b}(\om),\ \  \a\in \cZ^{2k}_{\geq} . \label{eqp2}
\end{eqnarray}

In order to write down the quantised equations in normal ordering (\ref{ixk}), 
we investigate the set.
$\cA^k\cap \cZ^k_{\geq}$.
We define a subset of of $\cA^k\cap \cZ^k_{\geq}$ (c.f. (\ref{seta}) and
(\ref{setz})), denoted by $\cN^k$:
\begin{equation}\label{setn}
 \cN^k=\left\{\a\in \cZ^k_{\geq} \cap \cA^k\big{|} \a_k=0 \right\},\qquad k\in 
\N,
\end{equation}
which is useful to write explicitly the Hamiltonians for the quantised Volterra 
hierarchy next section.
For any fixed $k\in \N$, all its elements can be constructed following the same
manner of Pascal's triangle. For example,
 $\cN^1=\{(0)\}$, $\cN^2=\{(1,0), (0,0)\}$ and
\begin{eqnarray}
&&\hspace{-1cm} \cN^3=\{(2,1,0), (1,1,0), (1,0,0), (0,0,0)\};\label{n3}\\
&& \hspace{-1cm} \cN^4=\{(3,2,1,0), (2,2,1,0), (2,1,1,0), (2,1,0,0),
(1,1,1,0), (1,1,0,0), (1,0,0,0),(0,0,0,0)\}.\label{ex4a}
\end{eqnarray}
In fact, the set $\cN^k$ is in bijection with the set
\[
 \cU^k=\left\{(\theta_1,\ldots,\theta_k) \big{|} \theta_k=0,\ 
\theta_s\in\{0,1\},\ s=1, \cdots, k-1
\right\}.
\]
Elements of   $\cU^k$ are sequences of zeros and ones of length $k$ with the 
last element $\theta_k=0$. The set $\cU^k$ has $2^{k-1}$ elements. The 
bijection 
with $\cN^k$ is given by the invertible linear transformation
\begin{eqnarray}\label{bij}
(\a_1,\ldots \a_k)=(\theta_1,\ldots,\theta_k)\cB,\quad
\cB_{ij}=\left\{\begin{array}{ll}
0 &\mbox{if }\ i<j\\
1  &\mbox{if }\ i\geqslant j
\end{array}\right. ,
\end{eqnarray}
or simply $\a_m=\sum\limits_{n=m}^k\theta_n$. Thus, we can rewrite the set 
$\cN^k$ given by \eqref{setn} as
\begin{eqnarray}\label{setnn}
\cN^k=\{\alpha=(\a_1, \a_2 \cdots, \a_{k-1}, 0)\in\bbbz^k\Big{|}\ 
\a_i=\sum_{s=i}^{k-1} \theta_s, \ \theta_s\in\{0,1\}; i=1, \cdots,
k-1\}.
\end{eqnarray}
\begin{Pro}\label{prop2}
The cardinality of set $\cN^{k+1}$ is $2^{k}$, and set
$\cA^{k+1} \cap \cZ^{k+1}_{\geq} $ has a cardinality of $(k+2) 2^{k-1} $, $0\leq
k\in \Z$.
\end{Pro}
\begin{proof} Due to the bijection \eqref{bij}, the first part of the statement 
is obvious. To prove the second half, we
define a subset of $\cN^{k+1}$ as $\cN_{(j)}^{k+1}=\left\{\a\in 
\cN^{k+1}\big{|} 
0\leq \a_1=j\leq k  \right\}$,
whose cardinal number is $\binom{k}{j}$.
Note that $\cN^{k+1}=\cup_{j=0}^{k} \cN_{(j)}^{k+1}$ and there is no 
intersection among any subsets with different $j$.
The set $\cA^{k+1} \cap \cZ^{k+1}_{\geq} $ can be obtained from the subset
$\cN_{(j)}^{k+1}$: for any $\a\in \cN_{(j)}^{k+1}$, we can generate $j$ more 
distinct elements in the set,
namely, $\cS^{-l} \a\in \cA^{k+1} \cap \cZ^{k+1}_{\geq} \setminus \cN^{k+1}$ for
$1\leq l\leq \a_1$.
Thus its cardinality is
$$
\sum_{j=0}^{k} (j+1) \binom{k}{j}=\sum_{j=0}^{k} \binom{k}{j}+\sum_{j=1}^{k} j 
\binom{k}{j}=2^k+k 2^{k-1}=(k+2) 2^{k-1}
$$
as stated in the proposition.
\end{proof}
Combining Proposition \ref{prop1} and the construction of the set 
$\cA^{k+1} \cap \cZ^{k+1}_{\geq}$ described in Proposition \ref{prop2}, we are
able to explicitly write down
the expressions of $X^{(k)}$ in the quantum algebras:
\begin{theorem}\label{thqa}
Let $\fI$ be either $\fI_a$ or $\fI_b$. Then
\begin{eqnarray}\label{xka}
\pi_{\fI} (X^{(k)})=\sum_{\a\in \cN^k}\sum_{j=0}^{\a_1} 
P^{\fI}_{\a-\a_1+j}(\omega)
u_{\a-\a_1+j}.
\end{eqnarray}
\end{theorem}
Using this theorem and \eqref{volh}, we can explicitly write down the quantum 
Volterra hierarchy. In the quantum algebra $\fA_{\fI_a}$,
for $k\in\N$, we have
\begin{eqnarray}
&&\d_{t_k}(u)=\sum_{\a\in \cN^k}\sum_{j=0}^{\a_1} P^{\fI_a}_{\a-\a_1+j}(\omega)\ 
\pi_{\fI_a} \left(
u_{\a-\a_1+j+1} u -u u_{\a-\a_1+j-1}\right)\nonumber\\
 &&\quad=\sum_{\a\in \cN^k}\sum_{j=0}^{\a_1} P^{\fI_a}_{\a-\a_1+j}(\omega) 
\left(\omega^{\nu(\a-\a_1+j,-2)}
u_{\overline{\a-\a_1+j+1,0}} 
-\omega^{\nu(\a-\a_1+j,2)}u_{\overline{0,\a-\a_1+j-1}}\right),\label{qeqa}
\end{eqnarray}
where the notation $u_{\overline{\beta}}$ stands for the standard normally ordered monomial which is similar to $u_\beta$.

As an example, we work out the case when $k=3$. The elements in the set $\cN^3$
are listed in \eqref{n3}. According to \eqref{qeqa}, we have
\begin{eqnarray*}
&&  \d_{t_3}(u)=P^{\fI_a}_{(2,1,0)}(\omega) (u_3 u_2 u_1 u-\omega u_1 u^2 
u_{-1})
+P^{\fI_a}_{(1,0,-1)}(\omega) (u_2 u_1 u^2-u^2 u_{-1} u_{-2})\\
&&\qquad \quad+P^{\fI_a}_{(0,-1,-2)}(\omega) (\omega u_1 u^2 u_{-1}-u u_{-1} 
u_{-2} u_{-3})
+P^{\fI_a}_{(1,1,0)}(\omega) ( u_2^2 u_1 u- u^3u_{-1})\\
&&\qquad \quad+P^{\fI_a}_{(0,0,-1)}(\omega) (u_1^2u^2-u u_{-1}^2 u_{-2})
+P^{\fI_a}_{(1,0,0)}(\omega) ( u_2 u_1^2 u- u^2u_{-1}^2)\\
&&\qquad \quad+P^{\fI_a}_{(0,-1,-1)}(\omega) (u_1u^3-u u_{-1} u_{-2}^2)
+P^{\fI_a}_{(0,0,0)}(\omega) (u_1^3 u-u u_{-1}^3)\\
&&\qquad \quad= u_3 u_2 u_1 u
+u_2 u_1 u^2-u^2 u_{-1} u_{-2}
-u u_{-1} u_{-2} u_{-3}+  u_2^2 u_1 u- u^3u_{-1}\\
&&\qquad \quad
+(1+\omega) (u_1^2u^2-u u_{-1}^2 u_{-2})
+(1+\omega) ( u_2 u_1^2 u- u^2u_{-1}^2)+ u_1u^3-u u_{-1} u_{-2}^2
+ u_1^3 u-u u_{-1}^3 ,
\end{eqnarray*}
where we compute
$P_{\a}^{\fI_a}(\om)$
using \eqref{binom} in Proposition \ref{prop1}.

Although Theorem \ref{thqa} is valid for the quantisation ideal $\fI_b$,
to compute the quantum Volterra hierarchy in the quantum algebra $\fA_{\fI_b}$
is much harder since we don't have the similar result for $\fI_b$ as
the one in Proposition \ref{prop1} for $\fI_a$.

When $\om=1$, Gaussian binomial coefficients become the ordinary binomial
coefficients. Proposition \ref{prop1} gives the formulas $P_{\a}(1)$
for the commutative polynomials $X^{(k)}$ given by (\ref{xk}). It follows from 
(\ref{qeqa})
that the explicit formula for the whole hierarchy of symmetries of
the classical (commutative) Volterra lattice is
\begin{eqnarray}\label{hsymv}
 \d_{t_k}(u)=\sum_{\a\in \cN^k}\sum_{j=0}^{\a_1} P_{\a-\a_1+j}(1) \left(
u_{\a-\a_1+j+1} -u_{\a-\a_1+j-1}\right) u, \quad k\in \N.
\end{eqnarray}

\section{Quantum Hamiltonians for the quantised Volterra 
hierarchies}\label{sec3}
In the quantum theory we replace real valued commutative variables by self 
adjoint
operators with respect to some Hermitian conjugation $\dagger$.
 The Hermitian conjugation $\dagger$ in algebra  $\fA$ is defined by the
following rules
\[
 u_n^\dagger=u_n,\quad \alpha^\dagger=\bar{\alpha},\quad 
(a+b)^\dagger=a^\dagger+ b^\dagger,\quad (ab)^\dagger=b^\dagger 
a^\dagger,\qquad 
u_n,a,b\in\fA,\ \ \alpha\in\bbbc,
\]
where $\bar{\alpha}$ is the complex conjugate of $\alpha\in \bbbc$.
The Hermitian conjugation $\dagger$ can be extended to the quantum algebras
$\fA_{\fI_a}$ and $\fA_{\fI_b}$ by letting $\omega^\dagger=\omega^{-1}$.
 The quantisation ideals $\fI_a$ (\ref{idi}) and $\fI_b$ (\ref{idj})
 are $\dagger$--stable. We introduce the square root $q=e^{\frac12\i \hbar}$ of
$\omega=e^{\i \hbar}$
with  $\hbar\in\R$ a real constant (an analog
of the Plank constant).
The quantised Volterra hierarchy in the quantum algebra
$\fA_{\fI_a}$ is presented in the form \cite{CMW}
\begin{equation}\label{qvolh}
  u_{t_1}=q(u_1u-uu_{-1}),\qquad u_{t_\ell}= 
q^\ell\left(\cS(X^{(\ell)})u-u\cS^{-1}(X^{(\ell)})\right), \quad \ell\in\N.
\end{equation}
 As a consequence of the results shown later in this paper, all these
derivations are self-adjoint, which justifies the rescaling by $q^{\ell}$.

In the paper \cite{CMW}, we presented the Volterra lattice and its first
symmetry 
in
Heisenberg form
\begin{equation}\label{qvola}
\begin{array}{ll}
    \d_{t_1}(u_n)=\dfrac{1}{q^{-1}-q}[H_1,u_n],\qquad
&H_1=\sum\limits_{k\in\Z}u_k;\\
    \d_{t_2}(u_n)=\dfrac{1}{q^{-2}-q^2}[H_2,u_n],&H_2=
    \sum\limits_{k\in\Z}(u_k^2+u_{k+1}u_k+u_ku_{k+1}),
  \end{array}
\end{equation}
where $H_1$ and $H_2$ are self-adjoint, algebraically independent and commuting 
Hamiltonians in $\fA_{\fI_a}$.

In the quantum algebra $\fA_{\fI_b}$ with commutation relations (\ref{comm2})
we can also write the equation (\ref{voltf2}) in Heisenberg form
\begin{equation}\label{qvolb}
  \d_{t_2}(u_n)=\dfrac{1}{q^{-2}-q^2}[H_2,u_n].
\end{equation}
Note that in the quantum algebra $\fA_{\fI_b}$ we have $H_2=H_1^2$ and
$H_2^\dagger=H_2$.
In this section, we derive the explicit expressions for the self-adjoint, 
algebraically independent and commuting Hamiltonians of the Volterra hierarchy
in both quantised algebras  $\fA_{\fI_a}$ and $\fA_{\fI_b}$.

\subsection{Quantum Hamiltonians $H_n$ in $\fA_{\fI_a}$}
In section \ref{sec22}, we give the definition of the sets $\cA^k$, 
$\cZ^k_{\geq}$ and $\cN^k$, c.f. (\ref{seta}), (\ref{setz}) and \eqref{setn} (or 
equivalently \eqref{setnn}),
whose elements $\a$ are associated to the $k$-degree monomials $u_\a$ for 
$k\in\N$.
We now define another set related to them, namely,
 \begin{eqnarray*}
  \cM_j^k=\left\{\a\in \cZ^k_{\geq}\big{|}\ \exists\ i\in \{ 1, 2, \cdots, k\}\ 
\mbox{such that}\
\a_{i} = j \right\}.
 \end{eqnarray*}
Clearly, we have
$$ \cM_0^k=\cZ^k_{\geq} \cap \cA^k \quad \mbox{and} \quad \cN^k\subset 
\cM_0^k.$$
Note that the definition of the polynomials $P_{\a}(\om)$ for $\a \in \cM_0^k$ 
in (\ref{pa}), which can be extended to the set
 $\cZ^k_{\geq}$ by the convention $P_{\a}(\om) = 0$ if $\a \notin \cM_0^k$.

We first prove several lemmas. In the proofs of these lemmas, we drop the 
up-index and simply
write $P_{\a}(\om)$ for $P_{\a}^{\fI_a}(\om)$.

\begin{Lem}\label{lem1}
Let $\a\in \cZ^{\ell}_{\geq}$ be such that $\nu(\a,1)=\nu(\a,-1)$. Then 
$P_{\a+1}^{\fI_a}(\om) = P_{\a-1}^{\fI_a}(\om)$ and $u$ commutes with $u_{\a}$ 
in $\fA_{\fI_a}$,
i.e., $\pi_{\fI_a} \left([u, u_{\a}]\right)=0$.
\end{Lem}
\begin{proof} It is obvious that  $u$ commutes with such $u_{\a}$ in
$\fA_{\fI_a}$ due to the commutation relations (\ref{comm1}). We now prove
the rest of the statement by considering two cases. If $\nu(\a,1)=\nu(\a,-1) = 
0$, that is, $\a$ contains neither $-1$ nor $1$, this leads to $P_{\a+1}(\om) =
P_{\a-1}(\om)=0$
since neither $\a+1$ nor $\a-1$ is admissible (requiring to contain $0$).  If 
$\nu(\a,1)=\nu(\a,-1) \neq 0$,  we also have that $\nu(\a,0)\neq 0$ since 
$\a\in 
\cZ^{\ell}_{\geq}$.
From (\ref{eqp}) it follows that
\begin{eqnarray}\label{p1m1}
 P_{\a+1}(\omega)=\frac{\om^{\nu(\a,-1)}-1}{\om^{\nu(\a,0)}-1} P_{\a}(\omega); 
\quad P_{\a-1}(\omega)=\frac{\om^{\nu(\a,1)}-1}{\om^{\nu(\a,0)}-1} 
P_{\a}(\omega)
\end{eqnarray}
implying $P_{\a+1}(\omega)=P_{\a-1}(\omega)$.
\end{proof}

\begin{Lem}\label{lem2}
Let $Y^{(\ell)}= \cS(X^{(\ell)}) u -u \cS^{-1} (X^{(\ell)})$. Then in the 
quantum 
algebra $\fA_{\fI_a}$ we have
\begin{equation}\label{kll}
 Y^{(\ell)}= \sum_{\a \in \mathcal{N}^{\ell}} \sum_{k \in \bbbz} \frac{ 
P_{\a}(\omega)}{\om^{\nu(\a,0)}-1} [u, u_{\a+k}].
 \end{equation}
\end{Lem}
\begin{proof}
It follows from (\ref{pa}) that
$$\pi_{\fI_a} \left(\cS(X^{(\ell)})\right)=\sum_{\a \in \cM_0^{\ell}} 
P_{\a}(\omega) \, \,  u_{\a+1}.$$
 Thus we have
\begin{eqnarray}\label{pas1}
\pi_{\fI_a} \left(\cS(X^{(\ell)}) u\right)& =& \sum_{\a \in \cM_0^{\ell}} 
P_{\a}(\omega) \, \,  \pi_{\fI_a} (u_{\a+1} u) \nonumber \\
&=& \sum_{\a \in \cM^{\ell}_{0,(0 , -2)}} P_{\a}(\omega) \, \,  
\pi_{\fI_a}(u_{\a+1} u)+\sum_{\a \in \overline{\cM}^{\ell}_{0,(0, -2)}} 
P_{\a}(\omega) \, \,  \pi_{\fI_a}(u_{\a+1} u),
\end{eqnarray}
where we use the notations $\cM^{\ell}_{k,(i,j)}=\{\a \in \cM_k^{\ell}\big{|} 
\nu(\a,i)=\nu(\a,j)\}$ and 
$\overline{\cM}^{\ell}_{k,(i,j)}=\cM_k^{\ell}\setminus \cM^{\ell}_{k,(i,j)}$.
In the same way, we have
$$
\pi_{\fI_a} \left(u \cS^{-1}(X^{(\ell)}) \right)=
\sum_{\a \in \cM^{\ell}_{0,(0,2)}} P_{\a}(\omega) \, \, \pi_{\fI_a}(u u_{\a-1}) 
+\sum_{\a \in \overline{\cM}^{\ell}_{0,(0,2)}} P_{\a}(\omega) \, \,  
\pi_{\fI_a}(u u_{\a-1}).
$$
We claim that
\begin{equation*}
 \sum_{\a \in \mathcal{M}^{\ell}_{0,(0,2)}} P_{\a}(\omega) u_{\a-1}=  \sum_{\a 
\in \mathcal{M}^{\ell}_{0,(0,-2)}} P_{\a}(\omega)  u_{\a+1}
\end{equation*}
and that both sides commute with $u$ in $\fA_{\fI_a}$.
This is equivalent to
$$\sum_{\a \in \mathcal{M}^{\ell}_{-1,(-1,1)}} P_{\a+1}(\omega) u_{\a}=  
\sum_{\a \in \mathcal{M}^{\ell}_{1,(-1,1)}} P_{\a-1}(\omega)  u_{\a}, $$
which is indeed true using Lemma \ref{lem1}.

We then rewrite the rest of sums in (\ref{pas1})
as
\begin{eqnarray*}
&&\sum_{\a \in \overline{\cM}^{\ell}_{0,(0,-2)}} P_{\a}(\omega)  u_{\a+1}u =
\sum_{\a \in 
\overline{\cM}^{\ell}_{0,(0,-2)}}{\frac{P_{\a}(\omega)}{\om^{\nu(\a,0)-\nu(\a,-2
)}-1} [u, u_{\a+1}] }\\
&&\qquad=\sum_{\a \in \overline{\cM}^{\ell}_{1,(1,-1)}} 
{\frac{P_{\a-1}(\omega)}{\om^{\nu(\a,1)-\nu(\a,-1)}-1} [u, u_{\a}] }
=\sum_{\a \in \overline{\cM}^{\ell}_{1,(1,-1)}} {\frac{\om^{\nu(\a,-1)} 
P_{\a-1}(\omega)}{\om^{\nu(\a,1)}-\om^{\nu(\a,-1)}} [u, u_{\a}] } .
\end{eqnarray*}

Similarly, we are able to show that
\begin{equation*}
\sum_{\a \in \overline{\cM}^{\ell}_{(0,2)}}P_{\a}(\omega) u u_{\a-1}  = 
\sum_{\a 
\in \overline{\cM}^{\ell}_{-1,(-1,1)}} {\frac{\om^{\nu(\a,1)} 
P_{\a+1}(\omega)}{\om^{\nu(\a,1)}-\om^{\nu(\a,-1)}} [u, u_{\a}] } .
\end{equation*}
Taking the difference yields
\begin{equation} \label{eq}
    Y^{(\ell)}= \sum_{\a \in \overline{\cM}^{\ell}_{1,(-1,1)}} 
{\frac{\om^{\nu(\a,-1)} P_{\a-1}(\omega)}{\om^{\nu(\a,1)}-\om^{\nu(\a,-1)}} [u, 
u_{\a}] } -
    \sum_{\a \in \overline{\cM}^{\ell}_{-1,(-1,1)}} {\frac{\om^{\nu(\a,1)} 
P_{\a+1}(\omega)}{\om^{\nu(\a,1)}-\om^{\nu(\a,-1)}} [u, u_{\a}] } .
\end{equation}
We simplify it by splitting into different cases.
When $\nu(\a,1) \nu(\a,-1)\neq 0$, $\a$ belongs to both sums in \eqref{eq}. 
Using (\ref{p1m1}), the difference
of fractions simplifies into
$$
\frac{\om^{\nu(\a,-1)} P_{\a-1}(\omega)-\om^{\nu(\a,1)} 
P_{\a+1}(\omega)}{\om^{\nu(\a,1)}-\om^{\nu(\a,-1)}}=\frac{ 
P_{\a}(\omega)}{\om^{\nu(\a,0)}-1} u_{\a} .
$$
When $\nu(\a, -1) =0$ but $\nu(\a,0) \nu(\a,1) \neq 0$, $\a$
only appears in the first sum in \eqref{eq} and using (\ref{p1m1}) in that case 
one
can write
$$ 
\frac{P_{\a-1}(\omega)}{\omega^{\nu(\a,1)}-1}=\frac{P_{\a}(\omega)}{\omega^{\nu(
\a,0)}-1}.$$
Similarly, when $\nu(\a,1)=0$ but $\nu(\a,0) \nu(\a,-1) \neq 0$, we have
$$
\frac{P_{\a+1}(\omega)}{\omega^{\nu(\a,-1)}-1}=\frac{P_{\a}(\omega)}{\omega^{\nu
(\a,0)}-1}.$$
Finally, if $\nu(\a,-1)= \nu(\a,0)=0$ but $\nu(\a,1) \neq 0$, then $\a$ appears 
in the
first sum only in \eqref{eq} and we can rewrite the term as
$\frac{P_{\beta}(\omega)}{\omega^{\nu(\beta,0)}-1} u_{\beta+1},$
where $\beta$ is an element of $\mathcal{N}^{\ell}$, that is, to say 
$\beta_l=0$.

In the mirror case where $\nu(\a,1)= \nu(\a,0)=0$ but $\nu(\a,-1) \neq 0$, then 
$\a$
appears in the second sum in \eqref{eq} and we can rewrite the term as
$\frac{P_{\gamma}(\omega)}{\omega^{\nu(\gamma,0)}-1} u_{\gamma-1},$
where $\gamma \in \mathcal{M}_0^k$ is such that $\gamma_1=0$.

Thus we have so far proved that
$$Y^{\ell}= \sum_{\a \in \overline{\mathcal{M}}^{\ell}_{0,(1,-1)}} 
\frac{P_{\a}(\omega)}{\omega^{\nu(\a,0)}-1}[u,u_{\a}]+ \sum_{\beta \in 
\mathcal{N}^{\ell}} 
\frac{P_{\beta}(\omega)}{\omega^{\nu(\beta,0)}-1}[u,u_{\beta+1}]+ \sum_{\gamma 
\in \mathcal{M}^{\ell}_{0}, \gamma_{1}=0} 
\frac{P_{\gamma}(\omega)}{\omega^{\nu(\gamma,0)}-1}[u,u_{\gamma-1}].$$
Since elements $u_{\a}$ for $\a \in \mathcal{M}^{\ell}_{0,(1,-1)}$ commute with 
$u$ in $\fA_{\fI_a}$ following from Lemma \ref{lem1}, one can add them to the 
first sum
in the above formula and it becomes
\begin{equation} \label{eq2}
Y^{\ell}= \sum_{\a \in \mathcal{M}^{\ell}_{0}} 
\frac{P_{\a}(\omega)}{\omega^{\nu(\a,0)}-1}[u,u_{\a}]+ \sum_{\beta \in 
\mathcal{N}^{\ell}} 
\frac{P_{\beta}(\omega)}{\omega^{\nu(\beta,0)}-1}[u,u_{\beta+1}]+ \sum_{\gamma 
\in \mathcal{M}^{\ell}_{0}, \gamma_{1}=0} 
\frac{P_{\gamma}(\omega)}{\omega^{\nu(\gamma,0)}-1}[u,u_{\gamma-1}].
\end{equation}
Recursively applying (\ref{eqp}), we get $\frac{
P_{\a+m}(\omega)}{\om^{\nu(\a+m,0)}-1}=\frac{ 
P_{\a}(\omega)}{\om^{\nu(\a,0)}-1} 
$ when $\a\in \cM_0^{\ell}$ and $\a+m \in \cM_0^{\ell}$ for some $m\in \bbbz$.
Hence, we can rewrite \eqref{eq2} as
$$Y^{\ell}= \sum_{\a \in \mathcal{N}^{\ell}} \sum_{k=-\a_1-1}^1 \frac{ 
P_{\a}(\omega)}{\om^{\nu(\a,0)}-1} [u, u_{\a+k}].$$
Finally, adding extra shifts of $u_{\a}$ does not change the sum as these 
commute with $u$, and thus we complete the proof of the statement.
\end{proof}
\begin{theorem}\label{thm1}
The quantum Volterra hierarchy (\ref{qvolh}) in the algebra $\fA_{\fI_a}$ is
presented in the Heisenberg form:
\begin{equation}\label{qah}
 \d_{t_\ell}(u_n)=\frac{\i}{2\sin(\frac12\ell\hbar)}[H_\ell, u_n]
\end{equation}
where the Hamiltonians 
\begin{equation} \label{firstham}
        H_{\ell} = \sum_{\a \in \mathcal{N}^{\ell}} \sum_{k \in \bbbz} 
P_{\a}^{\fI_a}(\omega) \frac{\om^{\ell}-1}{\om^{\nu(\a,0)}-1} u_{\a+k},\quad 
\omega=q^2=e^{\i \hbar}, \qquad
\ell\in\bbbn
    \end{equation}
are self-adjoint and commute with each other.
\end{theorem}
\begin{proof}
The first part of the statement follows immediately from the definition of the 
quantised hierarchy \eqref{qvolh} and Lemma \ref{lem2}. We now show all 
Hamiltonians $H_{\ell}$ are self-adjoint. First note that, for any nonnegative 
integers $a$ and $b$,
$$
\binom{a+b-1}{a}_{\om^{-1}}=\om^{a(1-b)}\binom{a+b-1}{a}_{\om} .
$$
For $\a \in \cN^{\ell}$,  we have $\ell=\k_{\a_1}+\k_{\a_1-1}+\cdots+\k_1+\k_0$ 
and $\pi_{\fI_a} 
\left({u_\a}^{\dagger}\right)=\om^{\left(\k_{\a_1}\k_{\a_1-1}+\cdots+\k_2\k_1+\
k_1\k_0\right)} u_\a$.
Thus it follows from Proposition \ref{prop1} that
\begin{eqnarray*}
P_{\a}^{\fI_a}(\om^{-1})= 
\om^{\ell-\k_0-\left(\k_{\a_1}\k_{\a_1-1}+\cdots+\k_2\k_1+\k_1\k_0\right)} 
P_{\a}^{\fI_a}(\om), \quad \a \in \cN^{\ell}
\end{eqnarray*}
and hence
\begin{eqnarray*}
 H_{\ell}^\dagger=\sum_{\a \in \cN^{\ell}} \frac{\omega^{-\ell}-1
}{\om^{-\nu(\a,0)}-1} P_{\a}^{\fI_a}(\om^{-1})
 \sum_{k \in \mathbb{Z}}\cS^k \pi_{\fI_a} \left({u_\a}^\dagger\right)=H_{\ell}.
\end{eqnarray*}

Finally, we show that the Hamiltonians commute with each other. Let $\ell_1$ 
and 
$\ell_2$ be two positive integers and
define $Q=[H_{\ell_1}, H_{\ell_2}]$. Then $Q$ is of the form 
$$Q=\sum_{k\in\mathbb{Z}} \sum_{\a \in \cN^{\ell_1+\ell_2}}  
T_{\a}(\omega)u_{\a+k}$$
for some fractions $T_{\a}(\omega)$.
We know from \cite{CMW} that $[\d_{t_{\ell_1}}, \d_{t_{\ell_2}}]=0$. Hence for 
any $f\in \fA_{\fI_a}$, we have
$$
[H_{\ell_1}, [H_{\ell_2}, f]]-[H_{\ell_2}, [H_{\ell_1}, f]]=[f,[H_{\ell_1}, 
H_{\ell_2}]]=[f, Q]=0.
$$
Thus, if $Q\ne 0$, then every monomial of $Q$   belongs to the
center of $\fA_{J_a}$, which is impossible (see the proof of 
Proposition \ref{Za}).
\end{proof}

We apply Theorem \ref{thm1} to find the Hamiltonians for lower numbers $\ell$. 
When $\ell=1$, we know $\cN^1=\{(0)\}$ leading to
$$H_1=\sum_{k\in\mathbb{Z}}{u_k}.$$
When $\ell=2$, there are two elements in  $\cN^2$, namely, $(1,0)$ and 
$(0,0)$. It follows from (\ref{paa}) (or using Proposition \ref{prop1}) that
$P_{(1,0)}(\om) = 1$ and
$P_{(0,0)}(\om) = 1$.
Thus $$H_2=\sum_{k\in\mathbb{Z}}{u_k^2}+(\om+1) \sum_{k\in\mathbb{Z}}{u_{k+1} 
u_k}. $$ These are the same as those given by (\ref{qvola}).

When $\ell=3$, the set $\cN^3$ is given by \eqref{n3} and we have
$$P_{(2,1,0)}(\om) = 1,\quad P_{(1,1,0)}(\om)=1, \quad  P_{(1,0,0)}(\om)=1+\om, 
\quad P_{(0,0,0)}(\om)=1.$$
Hence
$$H_3=\sum_{k\in\mathbb{Z}}{u_k^3}+(\om^2+\om+1) 
\sum_{k\in\mathbb{Z}}\left(u_{k+2}u_{k+1}u_k+u_{k+1}^2u_k+u_{k+1}u_k^2\right). 
$$

For the quintic symmetry of the Volterra equation in $\fA_{\fI_a}$, that is,
$\ell=4$, the cardinality of $\cN^4$ is $8$ and whose elements is given in
\eqref{ex4a}.
Following (\ref{paa}) or using Proposition \ref{prop1}, we get
\begin{eqnarray*}
 && 
P_{(3,2,1,0)}(\om)=P_{(2,2,1,0)}(\om)=P_{(1,1,1,0)}(\om)=P_{(0,0,0,0)}(\om)=1;\\
 &&P_{(2,1,1,0)}(\om)=P_{(2,1,0,0)}(\om)=1+\om;
 \quad P_{(1,1,0,0)}(\om)=P_{(1,0,0,0)}=1+\om+\om^2.
\end{eqnarray*}
Thus using Theorem \ref{thm1} we obtain
\begin{equation*}
 \begin{split}
H_4 &=  \sum_{\mathbb{Z}}{u_k^4}+ (\om^2+1)(\omega+1)^2 
\sum_{\mathbb{Z}}{u_{k+2}u_{k+1}^2u_k}+ (1+ \omega+\omega^2)(\omega^2+1) 
\sum_{\mathbb{Z}}{u_{k+1}^2u_k^2} \\ &
+ (\om^2+1)(\omega+1) \sum_{\mathbb{Z}}{(u_{k+3}u_{k+2}u_{k+1}u_k + 
u_{k+2}^2u_{k+1}u_k+u_{k+2}u_{k+1}u_k^2+u_{k+1}^3u_k+ u_{k+1}u_k^3)} .
\end{split}
\end{equation*}

\subsection{Quantum Hamiltonians in $\fA_{\fI_b}$}
We have shown in \cite{CMW} that the derivations with even index $\ell$ in the 
nonabelian Volterra hierarchy stabilize the ideal $\fI_b$.
In this quantisation we can also find the Hamiltonians following the lines of the proofs in the previous section, except
that we now use the identity \eqref{eqp2} instead of \eqref{eqp}.
Because  of the similarity we will omit the proof and simply state the result with
examples.
\begin{theorem}\label{thm2}
The quantum Volterra hierarchy (\ref{qvolh}) in the quantum algebra 
$\fA_{\fI_b}$ is
presented in the Heisenberg form:
\begin{equation}\label{qbh}
 \d_{t_{2\ell}}(u_n)=\frac{\i}{2\sin(\ell \hbar)}[\hat{H}_{2\ell}, u_n],\ \
\hat{H}_{2\ell} = \sum_{\a \in \cN^{2\ell}} \sum_{k \in \mathbb{Z}}
 \frac{\omega^{2\ell}-1 }{{((-1)^{k}
\om)^{\nu(\a,0)}-1}} P_{\a}^{\fI_b}({ (-1)^{k}
\om})
 \cS^{k} u_\a ,
\end{equation}
where $\omega=q^2=e^{\i \hbar}, \hbar\in\R$ and $\ell\in \N$.
Moreover, all Hamiltonians are self-adjoint and commute with each other.
\end{theorem}
We apply Theorem \ref{thm2} to write down the quantum Hamiltonians for the 
cubic 
and quintic members  of the Volterra hierarchy in $\fA_{\fI_b}$.
\begin{Ex}
The cubic symmetry of the Volterra equation corresponds to $\ell=1$ in Theorem 
\ref{thm2}. We have $\cN^2=\{(1,0), (0,0)\}$. It follows from (\ref{pab}) that 
$P_{(1,0)}^{\fI_b}(\om) = 1$ and $P_{(0,0)}^{\fI_b}(\om) = 1.$
Thus
$$\hat{H}_2 =  \sum_{\mathbb{Z}}{u_k^2}+ \sum_{\mathbb{Z}} 
(1+(-1)^{k}\omega){u_{k+1}u_{k}},$$
which is the same as those given by (\ref{qvolb}).
\end{Ex}
\begin{Ex} When $\ell=2$ in Theorem \ref{thm2}, the set $\cN^4$ is given in 
(\ref{ex4a}). From (\ref{pab}) we get
\begin{eqnarray*}
 && 
P_{(3,2,1,0)}(\om)=P_{(2,2,1,0)}(\om)=P_{(1,1,1,0)}(\om)=P_{(0,0,0,0)}(\om)=1;\\
 &&P_{(2,1,1,0)}(\om)=P_{(2,1,0,0)}(\om)=1-\om;
 \quad P_{((1,1,0,0)}(\om)=P_{(1,0,0,0)}=1+\om+\om^2.
\end{eqnarray*}
Thus, using Theorem \ref{thm2} we obtain
\begin{equation*}
\begin{split}
\hat{H}_4 &= \sum_{\mathbb{Z}}{u_k^4} -(\om^4-1) 
\sum_{\mathbb{Z}}{u_{k+2}u_{k+1}^2u_k} + 
\sum_{\mathbb{Z}}{(\om^2+1)(\omega^2+(-1)^{k}\omega+1)
u_{k+1}^2u_k^2}\\
&+ \sum_{\mathbb{Z}} (\om^2+1)(1+(-1)^{k}\om)
\left(u_{k+3}u_{k+2}u_{k+1}u_k+u_{k+2}^2u_{k+1}u_k+u_{k+1}^3u_k+u_{k+1}u_k^3+u_{
k+3} u_{k+2} u_{k+1}^2\right).
\end{split}
\end{equation*}
\end{Ex}

\subsection{Periodic quantum Volterra system}

The infinite Volterra hierarchy admits a periodic reduction
$u_{n+M}=u_n$ for any integer period $M$. The periodic reduction can be
obtained by taking a quotient of the algebra $\fA$ over the ideal $\cI_M=\langle
\{u_{n+M}-u_n\}_{n\in\bbbz}\rangle$. The ideal $\cI_M$ is obviously 
$\d_{t_\ell}$--stable. We denote the quotient algebra $\fA_M=\fA\diagup\cI_M\simeq
\bbbc\langle u_1,\ldots ,u_M\rangle$. The $M$--periodic Volterra 
system and its symmetries are the sets of $M$ equations of the form (\ref{vol}) 
and (\ref{hvol}), where the index $n\in \bbbz_M =\bbbz\diagup 
M\bbbz$. 

The quantisation ideal $\fI_a$ is $\cS$--stable, that is, $\cS(\fI_a)=\fI_a$. Thus the
quantum algebra $\fA_{\fI_a}$ admits a periodic reduction for any $M>2$ and
$\fA_M^a=\fA_{\fI_a}\diagup \cI_M$ is isomorphic to $\bbbc\langle u_1,\ldots
,u_M\rangle\diagup \fI_a^M$, where
\[
  \fI_a^M=\langle u_M u_{1}-\omega u_{1}u_M,\  u_n u_{n+1}-\omega u_{n+1}u_n,\ 
\ u_n u_m-u_m u_n\ ;\ 1\leqslant n<m\leqslant M,\ 1<m-n<M-1 \rangle .
\]
In contrast to the infinite case, the algebra $\fA_M^a$ has a nontrivial center
$Z(\fA_M^a)$.
\begin{Pro}\cite{MV}\label{Za}
 If $M$ is odd then   $Z(\fA_M^a)=\bbbc[C]$, where $C=u_M
u_{M-1}\cdots u_1$.
If $M$ is even, then    $Z(\fA_M^a)=\bbbc[C_1,C_2]$ where
$C_1=u_{M-1} u_{M-3}\cdots u_1$ 
and $C_2=u_M u_{M-2}\cdots u_2$.
\end{Pro}
\begin{proof}{\cite{MV}} We consider the monomials $u_M^{i_M}\cdots u_2^{i_2}u_1^{i_1},$
where
the powers $i_n$ are nonnegative integers, as a basis for the algebra
$\fA_M^a$. Since
the commutation relations in $\fA_M^a$ are homogeneous, the center is
generated by monomials.  A monomial $u_M^{i_M}\cdots
u_2^{i_2}u_1^{i_1}$ belongs to the center if and only if
\[
 0=[u_M^{i_M}\cdots 
u_2^{i_2}u_1^{i_1},u_n]=(\omega^{i_{n-1}}-\omega^{i_{n+1}})u_M^{i_M}\cdots 
u_n^{i_n+1} u_2^{i_2}u_1^{i_1} \quad \mbox{for all $n\in\bbbz_M$}.
\]
Therefore $i_{n+2}=i_{n}$ for $n\in\bbbz_M$, which yields the claim due to the $M$--periodicity of the indices.
\end{proof}
The central elements are first integrals
(constants of motion) of the corresponding quantum periodic Volterra system. In
the periodic case equations of the quantum Volterra hierarchy can also be 
written in Heisenberg form (\ref{qah}) with the commuting self-adjoint
Hamiltonians
(\ref{firstham}) being the finite sums
\begin{equation} \label{firsthamM}
        H_{\ell} = \sum_{\a \in \mathcal{N}^{\ell}} \sum_{k \in \bbbz_M} 
P_{\a}(\omega) \frac{\om^{\ell}-1}{\om^{\nu(\a,0)}-1} u_{\a+k},\qquad 
\ell\in\bbbn.
    \end{equation}
In contrast to the infinite dimensional case,  only
$k=\lfloor\frac{M-1}{2}\rfloor$ first Hamiltonians are algebraically independent
first integrals and the rest are polynomials in the  first
integrals $H_1,\ldots, H_k$ and central elements of the algebra  $\fA_M^a$. 
Thus, periodic reductions of the Heisenberg equations (\ref{qah}) are integrable 
quantum systems, since
\[
\# \mbox{ commuting Hamiltonians }=\frac12 (M-\# \mbox{ generators of the 
center}).
\]
The cases $M=3$ and $M=4$ are superintegrable \cite{CMW}. The system of three 
(resp. four) equations admits two (resp. three) algebraically independent 
quantum first integrals.

For example, in the case $M=3$, the center of algebra $\fA_3^a$ is $C=u_3 u_2 u_1$.
There is only one commuting Hamiltonian, namely,
$H_1=u_1+u_2+u_3$.  The Hamiltonians $H_k, \ k\geqslant 2$ are polynomials in $C$ and $H_1$:
\[
 H_2=H_1^2,\quad H_3=H_1^3+3\omega^2 C,\quad H_4=H_1^4+4\omega^2(1+\omega) C 
H_1,\ \ldots\ .
\]
In the case $M=4$, the independent first integrals are
$$C_1=u_3 u_1,\quad C_2=u_4 u_2, \quad    H_1=u_1+u_2+u_3+u_4, $$
where $C_1, C_2$ are central elements of $\fA_4^a$. 
The Hamiltonians $H_k, \ k\geqslant 2$ are polynomials in these first integrals, namely,
\begin{eqnarray*}
 &&H_2=H_1^2-2 (C_1+C_2),\quad H_3=H_1^3-3(C_1+C_2)H_1,\\
&&H_4=H_1^4-4(C_1+C_2)H_1^2+4(1+\omega^2)C_1 C_2+2 (C_1^2+C_2^2),\ \ldots\ .
\end{eqnarray*}
In the case $M=5$, there are two commuting Hamiltonians, namely, $H_1$ and 
$H_2$, and we have
\[
 H_3=\frac32 H_2 H_1-\frac12 H_1^3,\quad
H_4=H_1^2 H_2+\frac12 H_2^2 -\frac12 H_1^4,\ \ldots\ .
\]
In the case $M=6$,  the independent first integrals are $C_1, C_2$ as well as $H_1$, $H_2$,
and
\[
 H_3=\frac32 H_2 H_1-\frac12 H_1^3+3(C_1+C_2),\quad
H_4=H_1^2 H_2+\frac12 H_2^2 -\frac12 H_1^4+4(C_1+C_2) H_1,\ \ldots\ .
\]

The quantisation ideal $\fI_b$ is $\cS^2$--stable. Thus the quantum algebra
$\fA_{\fI_b}$ admits a periodic reduction for any {\em even} $M=2N>2$. The 
algebra
$\fA_M^b=\fA_{\fI_b}\diagup \cI_M$ is isomorphic to $\bbbc\langle u_1,\ldots
,u_M\rangle\diagup \fI_b^M$, where
\[
  \fI_b^M=\langle u_M u_{1}-\omega u_{1}u_M,\  u_n u_{n+1}-(-1)^n\omega 
u_{n+1}u_n,\
\ u_n u_m+u_m u_n\ ;\ 1\leqslant n<m\leqslant M,\ 1<m-n<M-1 \rangle .
\]
\begin{Pro}\label{Prop9}\cite{MV} Let $M=2N$.
\begin{enumerate}
 \item[\rm (i)] If  $N$ is odd, then $Z(\fA_M^b)=\bbbc[C_1,C_2]$ where
\[ C_1=u_{M-1}
u_{M-3}\cdots u_1,\qquad C_2=u_{M}
u_{M-2}\cdots u_2.
\]
\item[\rm (ii)] If $N$ is even, then the center of   $\fA_M^b$ is generated by the
elements
\[\hat{C}_1=u_{M-1}^2 u_{M-3}^2\cdots u_1^2,\quad  \hat{C}_2=u_{M}^2
u_{M-2}^2\cdots u_2^2,\quad
\hat{C} =u_{M}u_{M-1}\cdots u_2 u_1,\]
where the generators $\hat{C}_1, \hat{C}_2$ and $\hat{C}$ are algebraically
dependent $\hat{C}_1\hat{C}_2=\omega^{M-2}\hat{C}^2$.
\end{enumerate}
\end{Pro}
\begin{proof} The proof is similar to the Proposition \ref{Za}. Details of the 
proof can be found in \cite{MV}
\end{proof}
\begin{Ex}
In the case $M=4$, the center of algebra $\fA_4^b$ is generated by
$\hat{C}_1=u_{3}^2  u_1^2,\  \hat{C}_2=u_{4}^2 u_2^2$ and $\hat{C} =u_{4}u_{3} 
u_2 u_1$.
The Hamiltonian of the  cubic member of the periodic Volterra
hierarchy is
 \[
 \hat{H}_2=\sum_{n=1}^4 u_k^2+(1-\omega)(u_2 u_1+u_4u_3)+(1+\omega)(u_3 
u_2+u_1u_4)=H_1^2,
 \]
where $H_1=u_1+u_2+u_3+u_4$. The elements $B_1=u_3u_1$ and $B_2=u_4u_2$ commute 
with each other, with $\hat{H}_2$ and anti-commute with $H_1$, that is,
\[
 [B_1,B_2]=0,\quad [B_1,\hat{H}_2]=0, \quad [B_2,\hat{H}_2]=0, \quad 
[B_1,H_1]_+=B_1 H_1+ H_1 B_1=0,\quad [B_2,H_1]_+=0.
\]
Moreover, the generators of the center $Z(\fA_4^b)$ can be
represented as
$$\hat{C}=\omega^{-1}B_2B_1=u_4 u_3 u_2 u_1,\quad  \hat{C}_1=-B_1^2=u_{3}^2  u_1^2,\quad  \hat{C}_2=-B_2^2=u_{4}^2 u_2^2.$$
In algebra $\fA_4^b$, we have the Hamiltonian $
\hat{H}_4=\hat{H}_2^2-2\hat{C}_1-2\hat{C}_2$.

Since the elements $B_1,B_2$ commute with the Hamiltonian $\hat{H}_2$ and are not
central, 
they can be regarded as Hamiltonians of commuting quantum  symmetries 
\[
\d_{\tau_1}( u_n)=[B_1,u_n]=2u_3u_1u_n,\qquad \d_{\tau_2}( 
u_n) =[B_2,u_n]=2u_4u_2u_n, \qquad n\in\Z_4,
\]
which is not possible in the commutative case.
\end{Ex}

The results in the above example can be generalised to the case when $M=2N$ and
$N$ is even. In $\fA_M^b$ the elements $B_1=u_{M-1}u_{M-3}\cdots u_1$ and  
$B_2=u_{M }u_{M-2}\cdots u_2$ satisfy the commutation relations
\begin{eqnarray*}
 &&[B_1,u_k]_+=[B_2,u_k]_+=0,\qquad\qquad \qquad k\in\bbbz_M,\\
&& [B_1,B_2]=[B_1,\hat{H}_{2\ell}]=[B_2,\hat{H}_{2\ell}]=0,\qquad \ell\in\bbbn,
\end{eqnarray*}
and the generators of the center $Z(\fA_M^b)$ can be represented as
\[
 \hat{C}_1=(-1)^{\frac{N(N-1)}{2}}B_1^2,\qquad \hat{C}_2=(-1)^{\frac{N(N-1)}{2}}B_2^2,\qquad \hat{C}=(-1)^{\frac{(N-1)(N-2)}{2}}\omega^{1-N}B_2B_1,
\]
where $\hat{C}_1, \hat{C}_2$ and $\hat{C}$ are same as in the second statement of Proposition \ref{Prop9}.
\section{Summary and Discussion}\label{discussion}

In this paper, we present explicit expressions for the infinite hierarchy of quantum Hamiltonians corresponding to both quantisations of the Volterra hierarchy, namely, quantisation ideals
$\fI_a$ and $\fI_b$, and show that they are self-adjoint and commute with each other. Moreover, the dynamical equations of the quantum hierarchy can be written in the Heisenberg form using these Hamiltonians. The proofs mainly rely on the explicit expressions of the Volterra hierarchy on a free associative algebra.

The quantum algebra $\fA_{\fI_a}[\om]$ can be regarded as a deformation of the 
commutative algebra $\tilde{\fA}=\fA_{\fI_a}[1]=\bbbc[u_n;n\in\bbbz]$. It is 
well known that taking the classical limit $\hbar\to 0$, and thus  
$\omega=e^{\i \hbar}\to 1$, one can equip $\tilde{\fA}$ with a Poisson algebra 
structure and turn the Heisenberg equations into the corresponding Hamiltonian 
ones \cite{dirac_book, PLV, Dirac25}. Let us denote by $\tilde{a}\in 
\tilde{\fA}$
the limit of $a\in\fA_{\fI_a}[\om]$, that is,
$\tilde{a}=\lim_{\hbar\to 0}a$. It is clear that for any $a,b\in 
\fA_{\fI_a}[\om]$ the commutator $[a,b]\in (\om-1)\fA_{\fI_a}[\om]$ and thus we 
can define the bracket
\begin{equation}\label{poisson_a}
 \{\tilde{a},\tilde{b}\}=\lim_{\hbar\to 0}\frac{1}{e^{\frac12\i 
\hbar}-e^{-\frac12\i \hbar}}[a,b].
\end{equation}
The bracket (\ref{poisson_a}) is a Poisson bracket on $\tilde{\fA}$. Indeed, it is $\bbbc$--bilinear and skew symmetric satisfying the Jacobi and Leibniz identities. Thus, we have
\begin{equation}\label{quadratic_p}
 \{  {u}_m,  {u}_n \}=(\delta_{m, n-1} -\delta_{m,n+1}) {u}_m  
{u}_n.
\end{equation}
It follows from Theorem \ref{thm1} that
\[
 \partial_{t_{\ell}}( {u}_n)=\lim_{\hbar\to 
0}\frac{1}{e^{-\frac12\i \ell\hbar}-e^{\frac12\i\ell \hbar}}[H_\ell,u_n]=\{ 
{u}_n,\tilde{H}_\ell\},
 \]
where
\[
\tilde{H}_\ell=\lim_{\hbar\to 0}\frac{1}{\ell}H_\ell= \sum_{k \in \bbbz}
 \sum_{\a \in \cN^{\ell}} \frac{ P_{\a}(1)}{ \nu(\a,0)} {u}_{\a+k}.
 \]
Hence the densities $\tilde{h}_\ell$ of the local conservation laws $\tilde 
H_\ell= \sum_{k \in \bbbz}\cS^k(\tilde{h}_\ell)$
for the Volterra hierarchy are given by
 \[\tilde{h}_\ell=\sum_{\a \in \cN^{\ell}} \frac{ P_{\a}(1)}{ \nu(\a,0)}  
{u}_{\a}=\sum_{\a \in \cN^{\ell}}\frac{1}{\kappa_{\a_0}}
\binom{\k_{\mx}+\k_{\mx-1}-1}{\k_{\mx}}...  \binom{\k_2+\k_1-1}{\k_2}
\binom{\k_1+\k_0-1}{\k_1}  {u}_{\a},
\]
where we recall that $\kappa_i=\nu(\a,i)$ is the number of $i$'s in $\a$. Indeed, we have
\begin{eqnarray*}
&&\tilde{h}_1=u, \quad \tilde{h}_2=\frac{u^2}{2}+ u_1 u, \quad \tilde{h}_3=\frac{u^3}{3}+u_1 u^2+u_1^2 u+u_2 u_1 u,\\
&&\tilde{h}_4=u_3 u_2 u_1 u+u_2^2 u_1 u+u_1^3 u
+\frac{u^4}{4}
+2 u_2 u_1^2 u+ u_2 u_1 u^2
+\frac{3}{2} u_1^2 u^2 +u_1 u^3, \ \cdots
\end{eqnarray*}
As a byproduct of our results, we obtained explicit expressions for all local classical
Hamiltonians $\tilde{h}_\ell$ of the classical  commutative Volterra hierarchy.
Traditional approaches \cite{InKa}, based on the Lax representation of the 
Volterra lattice or the transfer matrix approach, enable one to find a {\em
generating function} for the Hamiltonians, but not their explicit expressions.

The quantum algebra $\fA_{\fI_b}[\om]$ can be regarded as a deformation of the 
noncommutative algebra $\breve\fA=\fA_{\fI_b}[1]$. A classical 
Poisson algebra structure and Hamiltonian description of equations associated 
with deformations of noncommutative algebras have been recently developed in  
\cite{MV}.

In the classical theory of integrable systems with commutative variables and 
systems on free associative algebra,
there are numerous powerful tools and useful concepts, including Lax 
representations, Darboux transformations, recursion operators, and master 
symmetries \cite{cw19-2, Olver, miksok_CMP, Sok-20}. Their connections with 
the concept of quantisation ideals have not been explored yet. Developing this 
aspect of the theory will enable us to take advantage of  a wide range of 
results in integrable systems and to advance the  theory based on quantisation 
ideals.

\section*{Acknowledgments}
AVM and JPW are grateful for the partial support by the EPSRC  grant 
EP/V050451/1.
SC thanks for the National Research Foundation of Korea(NRF) grant funded by the Korea government(MSIT) (No.2020R1A5A1016126).
This article is partially based upon work from COST Action CaLISTA CA21109 supported by COST (European Cooperation in Science and Technology). www.cost.eu.

\end{document}